\documentclass[12pt]{llncs}
\usepackage{times,fullpage}
\usepackage{amsmath,amssymb, mathrsfs} 
\usepackage{xspace}
\usepackage{xcolor}

\definecolor{verde}{rgb}{0.,0.6,0.2}
\definecolor{bianco}{rgb}{1.,1.,1.}
\definecolor{marrone}{rgb}{0.7,0.2,0.1}
\definecolor{rosso}{rgb}{1,0,0}
\definecolor{giallo}{rgb}{1.0, 0.87, 0.0}
\definecolor{blu}{rgb}{0.03, 0.27, 0.49}
\definecolor{daffodil}{rgb}{0.03, 0.27, 0.49}
\definecolor{darkcerulean}{rgb}{1.0, 0.87, 0.0}

\newcommand{\w}{\ensuremath{\omega}\xspace}

\newcommand{\calH}{\ensuremath{\mathcal{H}}\xspace}
\newcommand{\pp}{\ensuremath{\mathbb {P}}\xspace}

\usepackage{tikz}

\usetikzlibrary{matrix}
\usetikzlibrary{positioning}
\usetikzlibrary{shapes.geometric}
\usetikzlibrary{shapes,decorations,shadows}  
\usetikzlibrary{decorations.pathmorphing}  
\usetikzlibrary{decorations.shapes}  
\usetikzlibrary{fadings}  
\usetikzlibrary{patterns}  
\usetikzlibrary{calc}  
\usetikzlibrary{automata}
\usetikzlibrary{chains,fit,shapes}
\usetikzlibrary{backgrounds}

\def\e{{\rm e}}

\newcommand{\margine}[1]{}

\begin{document}

\title{Vertex-Connectivity Measures for Node Failure Identification in Boolean Network Tomography}

\author{Nicola Galesi\inst{1} \and Fariba Ranjbar\inst{1} \and Michele Zito\inst{2}}

\institute{Universit\'a La Sapienza Roma \and University of Liverpool}

\maketitle

\begin{abstract}
  In this paper we study the node failure identification problem in undirected
  graphs by means of Boolean Network Tomography. We argue that vertex
  connectivity plays a central role.
  We show tight bounds on the maximal identifiability in a
  particular class of graphs, the Line of Sight networks.
  We prove slightly weaker bounds on arbitrary networks.
  Finally we initiate the study of maximal identifiability in random
  networks. We focus on two models: the classical Erd\H{o}s-R\'enyi model,  and that of Random Regular graphs. The framework proposed in the paper allows a probabilistic analysis of the identifiability in random networks giving a tradeoff between the number of monitors to place and the maximal identifiability.
\end{abstract}

\section{Introduction}
\label{intro}

A
\margine{Network Tomography as a way to check network reliability} 
central issue  in communication networks  is to ensure that the structure works 
reliably. To this end it is of the utmost importance to discover as
quickly as possible those components that develop some sort of failure. Network Tomography is a family of distributed failure detection algorithms
based on the spreading of end-to-end measurements
\cite{coates02:_inter_tomog,vardi96:_networ_tomog} rather than directly
measuring individual network components. Typically a
network $G = (V,E)$ is given as a graph along with a
collection of paths \pp in it and the goal is to take measurements along such
paths to infer properties of the given network.
Quoting from \cite{DBLP:conf/imc/Duffield03} ``A key advantage of tomographic methods is that they require no participation from network elements other than the usual forwarding of packets. This distinguishes them from well-known tools such as {\tt traceroute} and {\tt ping}, that require ICMP responses to function. In some networks, ICMP response has been restricted by administrators, presumably to prevent probing from external sources. Another feature of tomography is that probing and the recovery of probe data may be embedded within transport protocols, thus co-opting suitably enabled hosts to form impromptu measurement infrastructures''.
 The
\margine{Connection to group testing}
approach is strongly related to group testing
\cite{du00:_combin_group_testin_its_applic} where, in general, one is
interested in making statements about individuals in a population by
taking group measurements. The main concern is to do so with the minimum
number of  tests.
In our setting, the connectivity
structure of the network constrains the set of feasible
tests. Graph-constrained group testing has been studied before,
starting with \cite{cheraghchi12:_graph}. We   are interested in using structural 
graph-theoretic properties to make statements about
the quality of the testing process.

Research
\margine{Boolean Network Tomography}
in Network Tomography is vast. The seminal works of Vardi \cite{vardi96:_networ_tomog}, and
Coates \cite{coates02:_inter_tomog},  or more recent surveys like
\cite{castro04:_networ_tomog} each have more that 500 citations, according to Google scholar. Methods and algorithms
vary dramatically depending on the network property of interest, or the measurements one has to rely on.
Boolean Network Tomography  (BNT) aims to identify corrupted components in a network 
using boolean measurements (i.e. assuming that elementary network components
can be in one of two states: ``working'' or ``not-working''). Introduced in \cite{DBLP:conf/imc/Duffield03,Ghita:2011:SNT:2079296.2079320},
the paradigm has recently attracted a lot of interest because of its simplicity.  
In this work  we  use BNT to identify failing  nodes.
Assume to have a set \pp of measurement paths over a node set $V$. We would like to know the state 
$x_v$  (with $x_v = 0$ corresponding to ``$v$ in working order'' and $x_v = 1$ corresponding to ``$v$ in a faulty state'') of each node $v \in V$. The  localization of the failing nodes in \pp  is captured by the solutions of the system: 
\begin{eqnarray}
\label{eqn:1}
\bigwedge_{p \in \pp} \left( \bigvee_{v\in p} x_v\equiv b_p \right )
\end{eqnarray}
where  $b_p$ models the (boolean) state  of the path $p \in \pp$.
Of course, systems of this form may have several solutions and
therefore, in general, the availability of a collection of end-to-end
measurements does not necessarily lead to the unique identification
of the failing nodes. We
\margine{Aim: study structural properties that guarantee good failure  detection capabilities}
will investigate properties of the underlying network that facilitate the solution of this problem.
In particular, we follow
the approach initiated by Ma et. al.
\cite{DBLP:conf/imc/MaHSTLL14} based on the
notion of {\em maximal identifiability} (see Section  \ref{sec:prel} for a precise definition).
 The metric aims to capture the maximal number of simultaneously
 failing nodes that can be uniquely identified in a network by means of
 measurement along a given path system.  It turns out that the network maximal identifiability is an interesting combinatorial measure and several studies
 \cite{BHK17,GR18,DBLP:conf/imc/MaHSTLL14,DBLP:conf/infocom/RenD16}
  have investigated variants of this measure in connection with
 various types of path systems. However, it seems difficult to come up
  with simple graph-theoretic properties that affect the given network
  identifiability. We contend  that working with the collection of
  simple paths between two disjoint sets of vertices $S$ and $T$ enables
  us to make good progress on this issue. 
More specifically
\margine{Result one: maximal identifiability of LoS networks}
we show that the proposed approach provides an almost tight characterization of
the maximal identifiability in {\em augmented hypergrids} (see
definition in Section \ref{sec:prel}) and more general Line-of-Sight (LoS) networks.
LoS networks were introduced by Frieze {\it et al.} in \cite{frieze07:_line_sight_networ} and have been  widely studied (see for instance
\cite{devroye13:_connec,DBLP:conf/saga/CzumajW07,sangha17:_indep_sets_restr_line_sight_networ,sangha17:_findin_large_indep_sets_line_sight_networ})
as models for communication patterns 
in a geometric environment containing obstacles. Like grids, LoS
networks can be embedded in a finite cube of
$\mathbb{Z}^d$, for some positive integer $d$. But LoS networks
generalize grids in that edges are allowed between nodes that are not necessarily next to each other in the network embedding.

Using the network vertex-connectivity, $\kappa(G)$, (i.e. the size of the minimal set of nodes disconnecting the graph) we are able to 
prove the following:

\begin{theorem}
  \label{los}
  Let $n$ be a positive integer, and $d$, and $\omega$ be fixed
  positive integers, independent of $n$.

  \begin{enumerate}
\item  The maximal identifiability
  of an augmented hypergrid  ${\cal H}_{n,d,\omega}$ on $n^d$
  vertices with range parameter $\omega$.
  is
  between $\kappa({\cal H}_{n,d,\omega})-1$ and $\kappa({\cal H}_{n,d,\omega})$.
  
\item   Let $G = (V,E)$ be an arbitrary $d$-dimensional LoS network
  with range parameter $\omega$. Then the maximal
  identifiability of $G$ is between $\kappa(G)-2$ and $\kappa(G)$.
\end{enumerate}
\end{theorem}

The result on
\margine{Result 2: maximal identifiability of arbitrary graphs}
LoS networks immediately suggests 
the related question about  general graphs.
In this work we prove upper and lower bounds on the maximal identifiability
of any network $G$. The following statement summarizes our findings:

\begin{theorem}
\label{gen}
  Let $G = (V,E)$ be an arbitrary graph. Then the maximal
  identifiability of $G$ is at least $\lfloor \kappa(G)/2 \rfloor -2$
  and at most $\kappa(G)$.
\end{theorem}
In both Theorem \ref{los} and \ref{gen}, the upper bound is proved by
showing that there are sets of $\kappa(G)+1$ vertices that cannot be identified.
The lower bounds which require the construction of paths separating large sets of
nodes in the graph,  is based on a well-known relationship between
$\kappa(G)$ and the existence of collections of vertex-disjoint paths
between certain sets of nodes in $G$.  In fact a much higher
lower bound can be proved for graphs with low connectivity (see
Theorem \ref{thm:kappa3} in Section \ref{imp}). The
result, which implies the aforementioned result for arbitrary LoS
networks, applies to many topologies studies in relation to
communication problems including various types of grids, butterflies,
hypercubes, and sensor networks.

Finally,
\margine{Result 3: random graphs}
we look at random networks (Erd\H{o}s-R\'enyi   and Random
Regular Graphs).  In these structures we are able to show a  trade-off between the
success probability of the relevant path construction processes  and
the size of the sets $S$ and $T$ defining the path set \pp.
Random graphs also give us alternative constructions of networks with large identifiability.

The rest of the paper is organized as follows. After
\margine{Organization of the paper}
a section devoted to preliminaries and  important definitions, we have a 
section that focuses on Theorem \ref{los}.1.
Section \ref{sec:vc} focuses on  arbitrary
graphs. First we look at the upper bound in Theorem \ref{los}.2. Then
an additional lower bound is proved for graphs with low connectivity,
which implies the lower bound in Theorem \ref{los}.2.  Finally 
Section \ref{sec:random} is dedicated to the analysis of the maximal
identifiability in random graphs. First we look at  Erd\H{o}s-R\'enyi
graphs, then random regular graphs.

\section{Preliminaries}
\label{sec:prel}

\paragraph{Sets, Graphs, Paths, and Connectivity.}
\margine{Sets, graphs, and paths}
If $U$ and $W$  are sets,
$ U\triangle W= (U \setminus W) \cup (W \setminus U)$ is the symmetric difference between $U$ and $W$.
Graphs (we will use the terms network and graph interchangeably)  in this paper will be undirected, simple and loop-less. 
A path  (of length $k$)  in a graph $G = (V,E)$  from  a node $u$
to a node $v$ is a sequence of nodes $p=u_1,u_2, \ldots, u_{k+1}$ such
that $u_1=u$, $u_{k+1}=v$ and $\{u_iu_{i+1}\} \in E$ for all
$i\in[k]$. The path $p$ is {\em simple} of no two $u_i$ and $u_j$ in
$p$ are the same. Any sub-sequence $u_x, \ldots, u_{x+y}$ ($x \in \{1,
\ldots,k+1\}$, $y \in \{0, \dots, k+1-x\}$) is said to be {\em
  contained in $p$}, and dually we say that $p$ {\em contains} the
sequence or {\em passes through} it. We say that path $p$ and $q$ {\em intersect} if
they contain a common sub-sequence. The intersection of a path $p$ and an arbitrary set of nodes $W$ is the set of elements of $W$ that are contained in $p$.
For a node $u$ in $G$, $N(u)$ is the set of {\em neighbourhood} of $u$,
i.e. $\{v \in V \;|\; \{u,v\}\in E\}$. The {\em degree} of $u$,  $\deg(u)$,
is the cardinality of $N(u)$, and let $\delta(G)=\min _{u\in V}\deg(u)$ be  the minimum degree of $G$.

In
\margine{Connectivity}
what follows $\kappa(G)$ denotes the vertex-connectivity of the
given graph $G = (V,E)$,  namely
$\kappa(G)$ is the size of the minimal subset $K$ of $V$,  such that removing $K$ 
from $G$ disconnects $G$. In particular it is well-known (see for
example \cite{Harary1969}, Theorem 5.1, pag 43) that
\begin{equation}
  \label{eq:2}
  \kappa(G) \leq \delta(G).
\end{equation}
It will also be convenient to work with sets of vertices disconnecting
particular parts of $G$. If $S, T \subseteq V$, then $\kappa_{ST}(G)$ is the size of the smallest vertex separator of $S$ and $T$ in $G$, i.e. the smallest set of vertices whose removal disconnects $S$ and $T$ (set $\kappa_{ST}(G) = \infty$ if
$S \cap T \neq \emptyset$ or there are $s \in S$ and $t \in T$ such that $\{s,t\} \in E$). Notice that $\kappa_{ST}(G)\geq \kappa(G)$.

\paragraph{Grids and LoS networks.}
For
positive integers $d$, and $n \geq 2$, let $\mathbb{Z}^d_{n}$ be the
$d$-dimensional cube  $\{1, \ldots, n\}^d$.
We say that distinct points $P_{1}$ and $P_{2}$ in a cube {\em share a line of
  sight} if their coordinates differ in a single place.
A graph $G=(V,E)$ is said
to be a {\em Line of Sight (LoS) network of size $n$, dimension $d$, and
range parameter $\omega$} if there exists an embedding $f_{G}:V \rightarrow
\mathbb{Z}^d_{n}$ such that $\{u,v\} \in E$ if and only if $f_{G}(u)$
and $f_{G}(v)$ share a line of sight and $|f_G(u)-f_{G}(v)| <\omega$. 
In the rest of the paper a LoS network $G$ is
always given along with some embedding $f_G$ in $\mathbb{Z}^d_{n}$ for some
$d$ and $n$, and with slight {\it abus de langage} we will often refer
to the vertices of $G$, $u, v  \in V$ in terms of their
corresponding points $f_G(u), f_G(v), \ldots$ in $\mathbb{Z}^d_{n}$,
and in fact the embedding $f_G$ will not be mentioned explicitly.
\begin{figure}[b]
\begin{center}
\includegraphics[scale=.4]{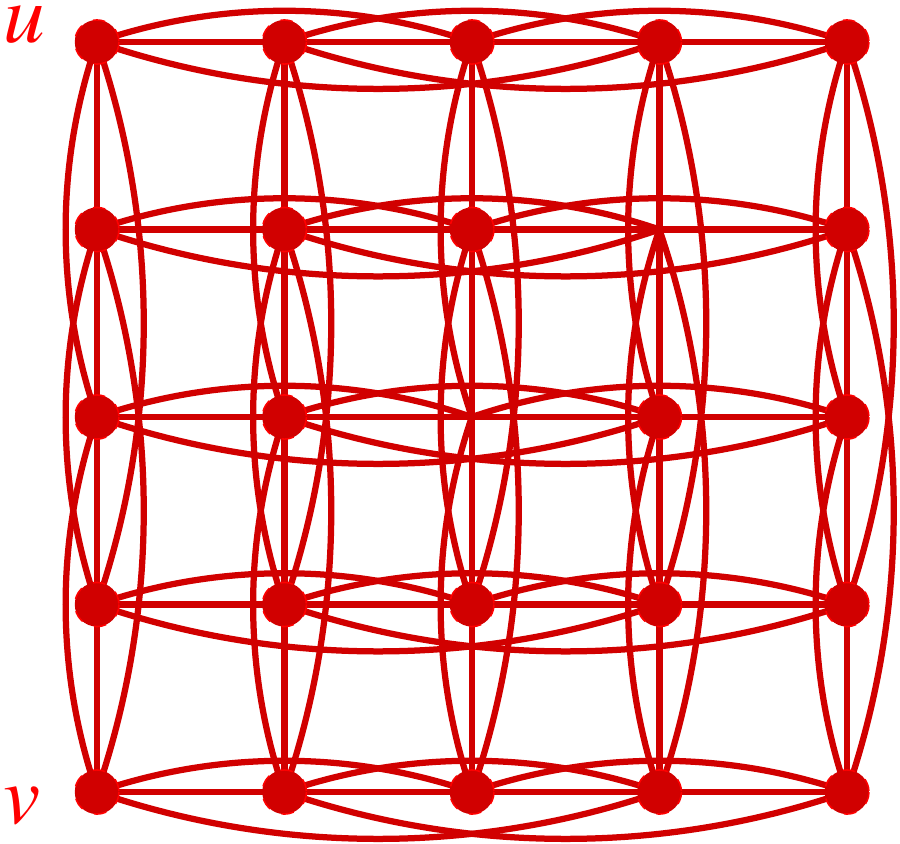}
\hspace{3cm}
\includegraphics[scale=.4]{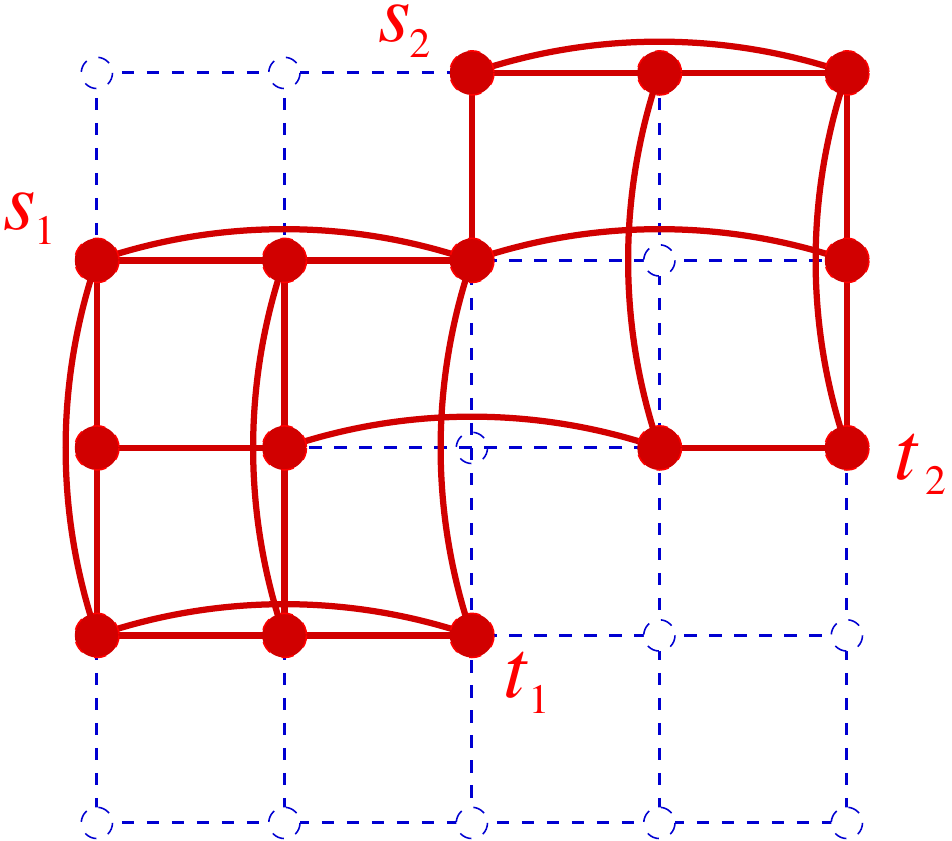}

(a) \hspace{5cm} (b)
  \caption{  \label{los-picture} On the left, the network
    $\calH_{n,\omega}$ for $n=5$ and $\omega=4$ (note that vertices
    $u$ and $v$ are not adjacent); on the right a more general
    example of  LoS  network, having $\omega = 3$, embedded in
    $\mathbb{Z}^2_{5}$ (represented as a dashed grid).}  
\end{center}
\end{figure}
Note that $d$-dimensional hypergrids, $\calH_{n,d}$, as defined in \cite{GR18} are particular LoS
networks with $\omega = 2$ and all possible $n^d$ vertices.
In the forthcoming sections we will 
study {\em augmented} hypergrids $\calH_{n,d,\omega}$ (or simply
$\calH_{n,\omega}$ in the 2-dimensional case), namely $d$-dimensional LoS networks
with range parameter $\omega>2$ containing all possible $n^d$ nodes.

\paragraph{Paths, Monitors and Identifiability.}
In BNT
\margine{Probing scheme}
one takes measurements along paths, and the quality of the
monitoring scheme depends on the choice of such paths.
Let $\pp$ be a set of paths over some node set $V$. For a node $v\in V$, let  $\pp(v)$ be the set of paths in \pp passing through $v$.  
For a set of nodes $U$, $\pp(U)=\bigcup _{u\in U}\pp(u)$. Hence if
$U\subseteq V$, $\pp(U)\subseteq \pp(V)$. Crucially, we
identify two disjoint sets of vertices $S$ and $T$, and  assume that  \pp
is the set of all simple paths in $G$ with one end-point in $S$ and
the other one in $T$. This is similar to the CSP probing scheme
analyzed in \cite{DBLP:journals/ton/MaHLST14}, but the scheme in that
paper does not assume $S \cap T = \emptyset$.

Traditionally
\margine{Monitors, where are they?}
in Network Tomography all measurements originate and end
at special {\em monitoring stations} that are connected to the structure under
observation. For any tomographic process to have any chance of
succeeding one has to assume that such {\em monitors} are
infallible. It is therefore customary to assume that the monitors are
external to the given network, but connected to it through a
designated set of nodes. $S \cup T$ is such set in our case. 
We call the pair $(S,T)$ a {\em monitor placement}.
In
\margine{Identifiability}
this settings, two sets of vertice $U$ and $W$ are {\em separable} if $\pp (U)\triangle  \pp (W) \not
=\emptyset$. 
A set of vertices $N$ is {\em $k$-identifiable} (with respect to the
probing scheme $(\pp,S,T)$) if and only if  
any  $U,W \subseteq N$, with $U\triangle W \not =\emptyset$ and
$|U|,|W|\leq k$,  $U$ are separable. The \emph{maximal identifiability} of $N$ with respect to 
$(\pp,S,T)$, $\mu(N,\pp,S,T)$,  is the largest $k$ such that $N$ is
$k$-identifiable.
For a graph $G=(V,E)$,
we write $\mu(G, \pp,S,T)$ to indicate the maximal identifiability of the set of nodes in $V$ which are used in
at least in a  path of $\pp$. In what follow we usually omit the
dependency of $\mu$ on the probing scheme $(\pp,S,T)$ when this is
clear from the context.

\smallskip

Note
\margine{Proof techniques}
that $k$-identifiability is monotone: if $G$ is
$k$-identifiable then it is  $k'$-identifiable for any
$k'<k$.
This implies that to prove that $\mu(N)\leq k-1$ it is sufficient to show that
$N$ is {\bf not} $k$-identifiable. By the definition given above this
boils down to showing the existence of two distinct node sets $U$ and
$W$ in $N$ of cardinality at most $k$ that are not separable.

Conversely, if we want to prove that $\mu(N)\geq k$ for some $k$, then
it is enough to argue that all distinct node sets $U$ and $W$ of
cardinality $|U|,|W| \le k$ are separable.
To prove this we have to show that for any two  distinct
node sets $U$ and $W$ of cardinality at most $k$ there exists
a path in \pp intersecting exactly one between $U$
and $W$.

\section{Failure Identifiability  in Augmented Hypergrids}
\label{sec:LoS}

Let $\omega >2$ be an integer. In this section we analyze the maximal
identifiability of augmented hypergrids. 
To maximize clarity, we provide full details for the special case of
$\calH_{n, \omega}$, the 2-dimensional augmented hypergrid,  and then
state the general result, leaving its proof to the Appendix.
In \cite{GR18}
\margine{Upper bound for $d=2$}
two of us showed that 
$\mu(G)\leq \delta(G)$ for any $(\pp,S,T)$.
In $\calH_{n, \omega}$  each node $u$ has $\w-1$ edges for each
one of the possible directions (north, south, east, west). Hence the minimal 
degree in $\calH_{n, \omega}$ is reached at the corner nodes and it is $2(\w-1)$. 
Thus $\mu(\calH_{n, \omega})\leq 2(\w-1)$ for any $(\pp,S,T)$. In the remainder of this section
we pair this up with a tight lower bound for a specific monitor placement.
Note that these results readily imply Theorem \ref{los}.1 as in
augmented hypegrids the vertex connectivity is actually equal to the
network's minimum degree.

We
\margine{Lower bound for $d=2$}
say that  nodes with coordinates  $(1,j)$ in $\calH_{n, \omega}$, for
some $j \in \{1, \ldots, n\}$,  are in the \emph{north
  border} of  $\calH_{n, \omega}$. Analogously we can define \emph{south, west and east} borders of 
$\calH_{n, \omega}$.
A \emph{canonical monitor placement} for  $\calH_{n, \omega}$ is a
pair $(S,T)$, such that
$|S|=|T|= 2\w-1$ and nodes in $S$ are chosen among the west and north borders of  $\calH_{n, \omega}$ and  $T$ among the the south and east borders of  $\calH_{n, \omega}$.

Given a node $u$ of $\calH_{n, \omega}$, identified as a pair $(i,j)\in \mathbb{Z}^2_{n}$, we define: 
\[SE(u)=\{(l,k): l\geq i \vee k\geq j\} \quad \mbox{and}\quad NW(u)=\{(l,k): l\leq i \vee k\leq j\}.\]
We are now ready to state the main result of this section.
\begin{theorem}
\label{full-los-2}
Let $n,\omega \in \mathbb N$, $n\geq 2$ and $\omega>2$. Let $(S,T)$ be the canonical monitor placement for $\calH_{n, \omega}$. Then $\mu(\calH_{n, \omega}) \geq 2(\w-1)-1$.

\end{theorem}
\begin{proof}
We have to prove that given two node sets $U$, and $W$ of cardinality at most $2(\w-1)-1$, with $U\triangle W \not =\emptyset$ we can build an $S$-$T$ path touching exactly one of them.  Given a node $u \in U \setminus W$, let $S(u)=NW(u)$, the nodes in the North-West region of $u$ and let 
$T(u)=SE(u)$, the nodes in South-East region of $u$.
Notice that $(1,1) \in S(u)$ and $(n,n)\in T(u)$ and $S(u) \cap T(u) =\emptyset$.  Since $|S|>2(\w-1)$ and $|W|\leq 2(\w-1)$, there is a node in $s \in S\setminus W$. Assume that $s=(1,1)$ (if $s\not =(1,1)$ is similar and give even better results). Similarly for $T$, assume that $(n,n) \not \in W$.  
Consider the following definition:

\begin{definition} 
Given a node $u \in \calH_{n, \omega}$, and a set of nodes $W$ in
$\calH_{n, \omega}$, 
we say that a direction $X$ (north, south, west, east)  is {\em $W$-saturated} on $u$ 
 if moving from $u$ on direction $X$ there is, right after $u$,  a consecutive block of $\w-1$ nodes in $W$.
\end{definition}

The following  claim define two disjoint paths $i_u$ in $S(u)$ from
$s$ to $u$  and $o_u$ in $T(u)$ form $u$ to $t$ not touching
$W$. Their concatenation hence defines a path joining $S$ to $T$
passing through $u$ and not touching $W$.
\qed\end{proof}

\begin{claim}
Let $u\in S(u)$. There is a path $i_u$ in $S(u)$ from $(1,1)$ to $u$ not touching $W$. There is a path $o_u$ in $T(u)$ from $u$ to $t$ not touching $W$.
\end{claim} 

\begin{proof}
We prove the first one since they are the same.
By induction on $S(u)$. If $|S(u)|=1$, then $u=s$ and we have done. If $|S(u)|>1$. Since $|W| \leq 2(\w-1)-1$, and since a direction is $W$-saturated only if a block of  $\w-1$ consecutive elements of  $W$  appear after $u$ on that direction, then there is at a least a direction $X$ between North and West  which is not $W$-saturated. Hence there is a node $u'\in S(u) \setminus W$ on direction $X$ from $u$ at distance less than $\w$. Hence there is an  edge  $\{u',u\} \in \calH_{n, \omega}$.  Since $S(u') \subset S(u)$ the inductive hypothesis give us a path $i_{u'}$ as required. Hence the path $i_u=i_{u'},u$ is as required.
\qed\end{proof}

Theorem
\margine{Extension to $d>2$}
\ref{full-los-2}  easily generalizes to $d$-dimensional augmented hypergrids.

\begin{theorem}
  \label{full-los-d}
Let $d,n,\omega \in \mathbb N$, $d,n\geq 2$ and $\omega>2$. Let
$(S,T)$ be the canonical monitor placement for $\calH_{n,
  d,\omega}$. Then $\mu(\calH_{n, d,\omega}) \geq d (\w-1)-1$.
\end{theorem}

\section{General Topologies}
\label{sec:vc}

We now look at the maximal identifiability in arbitrary
networks. Theorem \ref{gen} stated in Section \ref{intro}
will be a consequence of two independent results.
In
\margine{Upper Bound}
\cite{GR18} it was proved that   $\mu(G)\leq \delta(G)$, for any
monitor placement $(S,T)$. Here we show that $\mu(G)$ can be upper
bound in terms of 
$\kappa_{ST}$, the size of the minimal node set separating $S$
from $T$.

\begin{theorem}
\label{thm:ub}
Let $G=(V,E)$ be a graph and $(S,T)$ be a monitor placement. Then $\mu(G)  \leq \kappa_{ST}(G)$.
\end{theorem}
\begin{proof}
  If there is no vertex set in $G$ separating $S$ and $T$, 
  $\kappa_{ST}(G) = \infty$ and the result is trivial.
Let $K$ be the set witnessing the minimal separability of $S$ from $T$ in $G$. 
Hence $|K|=\kappa_{ST}(G)$. Let $N(K)$ be the set of nodes
neighbours of nodes in $K$ and notice this cannot be empty since $K$ is disconnecting $G$. Pick one $w\in N(K)$ and 
define $U:=K$ and $W:=U\cup\{w\}$. Clearly $\pp(U) \subseteq \pp(W)$. To see the opposite inclusion assume that there exists a 
path from $S$ to $T$ passing from $w$ but not touching $U=K$. Then $K$ is not separating $S$ from $T$ in $G$. Contradiction.
\qed\end{proof} 

Note that, while in general $\kappa_{ST}(G)$ may be larger than
$\delta(G)$, if $S$ and $T$ are separated by a set of $\kappa(G)$
vertices then, by inequality (\ref{eq:2}), the bound in Theorem
\ref{thm:ub} is at least as good as the minimum degree bound proved
earlier by the first two authors \cite{GR18}. This implies the upper
bound in Theorem \ref{gen}.

\medskip

Moving
\margine{Lower bound}
to lower bounds, in this section we prove the following:

\begin{theorem}
\label{thm:lb:gen}
Let $G=(V,E)$ and $(S,T)$ be a monitor placement for $G$. Then $\mu(G) \geq \min(\kappa(G),|S|,|T|)-2$.
 \end{theorem}

The lower bound in Theorem \ref{gen} can be derived easily from
Theorem \ref{thm:lb:gen}.  Let $K$ be a 
vertex separator in $G$ of size $\kappa(G)$, set $S^K$ to be the
first $\lfloor \kappa(G)/2 \rfloor$ elements of $K$ and $T^K = K
\setminus S^K$.  By Theorem \ref{thm:lb:gen} the maximal
identifiability of $G$ is at least $|S^K| - 2 = \lfloor \kappa(G)/2 \rfloor - 2$.

The proof of Theorem
\ref{thm:lb:gen} uses Menger's Theorem, a well-known result in graph theory (see \cite[Theorem 5.10, p. 48]{Harary1969} for its proof).

  \begin{theorem}(Menger's Theorem)
\label{thm:menger}
Let $G=(V,E)$  be a connected graph. Then $\kappa(G) \geq k$ if and only if each pair of nodes in $V$ is connected by at least $k$ node-disjoint paths in $G$.
\end{theorem} 

Menger's Theorem is central to the following Lemma which is used in
the proof of Theorem \ref{thm:lb:gen}.

\begin{lemma}
  \label{claim:mg-new} Let $G=(V,E)$. Let $W\subseteq V$ such that   $|W|\leq \kappa(G)-2$. 
Then any pair of vertices in $V \setminus W$ is connected by at least two vertex-disjoint simple paths not touching $W$.
\end{lemma}
\begin{proof}
 By Menger's Theorem, for any pair of nodes $u$ and $v$ in $V\setminus W$ there are at least $\kappa(G)$ vertex-disjoint paths from $u$ to $v$ in $G$. Call \pp the set of such paths. Since  $|W|\leq \kappa(G)-2$, then the nodes of $W$ can be  in at most $\kappa(G)-2$ of paths in \pp. Hence there are at least two paths in \pp not touching $W$.  
\qed\end{proof}

\paragraph{Proof of Theorem \ref{thm:lb:gen}.}  Let $G=(V,E)$ be an
undirected connected graph and $(S,T)$ be a monitor placement in
$G$. Note that without loss of generality that $\min(\kappa(G),|S|,|T|)>2$  (for otherwise
there is nothing to prove).

Assume first that
$|S| \geq \kappa(G)$ and $|T| \geq \kappa(G)$.  We claim that 
\[\mu(G) \geq \kappa(G)-2.\]
We show that for any distinct non-empty subsets $U$ and $W$ of $V$ of size at most $\kappa(G)-2$, there is a path in \pp touching exactly one between $U$ and $W$. 
 Given such $U$ and $W$, fix a node $u \in  U\triangle W$ and wlog $u\in U$.
 Since $|W| \leq \kappa(G)-2$ and $|S| \geq \kappa(G)$ there is at least a node in 
 $s \in S \setminus W$. By  the Claim  above
 applied to nodes $s$ and $u$ and to the set $W$, there are two vertex-disjoint simple paths $\pi^s_1,\pi^s_2$ from $s$ to $u$ not touching $W$. The same reasoning applied to $T$, guarantees the existence of a node $t \in T\setminus W$ and two vertex-disjoint paths $\pi^t_1,\pi^t_2$ from $u$ to $t$  not touching $W$.
If at least one between $\pi^s_1$, and $\pi^s_2$ only intersects one
of $\pi^t_1$, and $\pi^t_2$ at $u$ then the concatenation of such
paths is a (longer) simple path from $s$ to $t$ passing through $u$ and not
touching $W$.
Otherwise the concatenation of one between $\pi^s_1$, and $\pi^s_2$ with one
between $\pi^t_1$, and $\pi^t_2$ is a non simple path. In what follow
we show that the subgraph of $G$ induced by the four paths does
contain a simple path from $s$ to $t$ passing through $u$ and not
touching $W$. In the construction below we exploit the fact that
$\pi^s_1$, and $\pi^s_2$ (resp. $\pi^t_1$, and $\pi^t_2$) are simple and vertex disjoint. 
Let $p$ be a path from $s$ to $u$. Define an order on the nodes of $p$ as follows: $v\prec w$ if going from $v$ to $u$ we pass though $w$.  For $i,j\in \{1,2\}$, let $Z_{ij}$ be the set of nodes in $\pi^s_i\cap \pi^t_j$. Notice that $Z_{ij}=Z_{ji}$. Nodes in $Z_{ij}$ can be ordered according to $\prec$.
So let $z_{ij}$ be the minimal node in $Z_{ij}$ wrt $\prec$. Wlog let us say that $z_{1j}\prec z_{2j}$.
Observe hence that the subpath  $\pi^s_1[s\ldots z_{1j}]$ of $\pi^s_1$ going from $s$ to $z_{1j}$, before $z_{1j}$ is intersecting neither $\pi^t_1$ nor $\pi^t_2$. Hence the concatenation of the following three disjoint paths defines a simple path from $s$ to $t$ passing through $u$ avoiding $W$, hence a path in \pp with the required properties:
\begin{enumerate} 
\item $\pi^s_1[s\ldots z_{1j}]$, going form $s$ to $z_{1j}$;
\item $\pi^t_j[z_{1j}\ldots u]$ a sub path of $\pi^t_j$ going from $u$ to $z_{1j}$ and traversed in the other direction;
\item $\pi^t_{j\!\!\! \mod 2+1}$, the other path connecting $u$ to $t$.
\end{enumerate}

Now assume that least one $|S|$
or $|T|$ is less than $\kappa(G)$. Let $r=\min(|S|,|T|)-2$.  As before
we prove that for all distinct non-empty $U$ and $W$ subsets of $V$ of size at most $r$, there is a simple $S-T$ path in $G$, hence in \pp, touching exactly one between $U$ and $W$. Let 
$u \in  U\triangle W$ and wlog $u\in U$.  Notice that
$r+2=\min(|S|,|T|)$, then both $|S|\geq r+2$ and $|T|\geq r+2$.  Since
$|W|\leq r$, as before there are $s\in S\setminus W$ and $t \in T \setminus W$. Furthermore, since $\kappa(G) \geq \min(|S|,|T|)$, then  by previous observation on $|S|$ and $|T|$, $\kappa(G) \geq r+2$ and, since $|W|\leq r$, then  
$\kappa(G) - |W| \geq 2$, that is $|W|\leq \kappa(G)-2$.  As
in the previous case we can apply the Claim above once to $s$, $u$ and $W$ getting the vertex-disjoint paths $\pi^s_1$ and $\pi^s_2$  from $s$ to $u$, and once to $t$, $u$ and $W$ getting the vertex-disjoint paths $\pi^t_1$ and $\pi^t_2$  from $t$ to $u$. The proof then follows by the same steps as in the previous case. 
We then have proved that  if $|S|$  or $|T|$ are $\leq \kappa(G)$, then  
$\mu(G) \geq \min(|S|,|T|)-2$ and the proof of Theorem
\ref{thm:lb:gen} is complete. \qed

\subsection{Improved Bounds for Networks with Low Connectivity}
\label{imp}

We
\margine{Tighter lower bound for graphs of small connectivity}
complete this section investigating a different way to relate the
graph vertex connectivity to  $\mu(G)$. It is easy to see that, in general, the bounds in Theorem \ref{gen} are not
very tight, particularly when $\kappa(G)$  is large. However, if
$\kappa(G)$ is small, we can do better.
Theorem \ref{thm:kappa3} below in particular applies to LoS network
with constant range parameter, and readily gives the lower bound
promised in Theorem \ref{los}.2.

\begin{theorem}
\label{thm:kappa3}
 Let $G=(V,E)$, and $\kappa(G) \leq \frac{|V|}{3}$. There exists  a monitor placement for $G$ 
 such that $\kappa(G)-2 \leq \mu(G) \leq \kappa(G)$.
\end{theorem}
\begin{proof}
Assume $\kappa(G)=k$, and let $K$  be a minimal vertex separator in $G$.
Let $G^K_i=(V^K_i,E^K_i)$, $i\in [r_K]$ be the $r_K\geq 2$ connected components remaining in $G$ after removing $K$.  
Since $k\leq \frac{n}{3}$, then $2k \leq n-k$. Since $|V\setminus K|$
has $n-k$ nodes there are sufficient nodes 
in $V\setminus K$ to define $(S,T)$ with $|S| =|T|=k$  in such a way that the smallest among the   $V_i^K$'s
contains only element from $S$. This can be done as follows:  if the
smallest $V_i^K$'s has less than $k$ nodes, say $k-\ell$,  then assign
all its nodes to $S$. Use other components $G_j^K$'s (that will have
more than $k+\ell $ nodes)  to assign $\ell$ nodes to $S$ and $k$
other nodes to $T$. If the smallest $V_i^K$ has more than $k$ nodes,
choose $k$ among them and put them in $S$. Choose $k$ nodes in other
components and assign them to $T$.

We now prove that  $\mu(G) \leq \kappa(G)$. Let $G_i^K$ be the component where all the $S$-nodes  are assigned. Let $w$ be a node in $V^K_i \cap N(K)$.  This node has to exists since $G$ was connected and the removal of $K$ is disconnecting $G^K_i$ from $K$. Fix $U=K$ and $W=K\cup\{w\}$. We will show that $\pp(U) =  \pp(W)$. It sufficient to prove that $\pp(\{w\}) \subseteq  \pp(K)$, since clearly $\pp(U) \subseteq  \pp(W)$. Observe that no $S-T$ path $p$ in $G$ can live entirely inside $G^K_i$, i.e. have all of  its nodes in $V_i^K$. This is because at least one end-point (that in $T$) it is necessarily  missing in any path entirely living only in $G^K_i$.  Hence  a path touching $w$ is either entering or leaving $G^K_i$. But outside of $G^K_i$ $w$ is connected only to $K$, since otherwise $K$ would not be a minimal vertex separator. Hence it must be necessary that  $\pp(\{w\}) \subseteq  \pp(K)$. We have found $U,W$ of size $\leq \kappa(G)$ such that $\pp(U) =  \pp(W)$.  The upper bounds follows. 
The lower bound follows form Theorem \ref{thm:lb:gen} noticing that
$|S| = |T| = \kappa$.  \qed
\end{proof}

Arbitrary LoS networks have minimum degree, and hence also vertex
connectivity at most $2d(\omega-1)$. The next corollary follows
directly from Theorem \ref{thm:kappa3}.

\begin{corollary}
 Let $G$ be an arbitrary LoS network, with fixed range parameter
 $\omega$. Then $\mu(G) \geq \kappa(G)-2$. 
\end{corollary}

\section{Random Networks and Tradeoffs} 
\label{sec:random}

The main aim of this work is to characterize the
identifiability in terms of the vertex connectivity. In this section
we prove that tight results are possible in
random graphs. Also we show an interesting trade-off between the
success probability of the various random processes and the size of
the sets $S$ and $T$. Finally, random graphs give us constructions of networks
with large identifiability.

\subsection{Sub-Linear Separability in Erd\H{o}s-R\'enyi Graphs}
\label{gnp-sect}

 We start our investigation of the identifiability  of node
failures in random graphs by looking at the binomial model
$G(n,p)$, for fixed $p \leq 1/2$ (in this section only we follow the
traditional random graph jargon and use $p$ to  denote the graph edge
probability rather than  a generic path). The following equalities,
which hold
with probability approaching one as $n$ tends to infinity high
probability (that is {\em with high probability} (w.h.p.)),  are folklore:
\begin{equation}
\kappa(G(n,p))= \delta(G(n,p)) = np - o(n).
\label{eq:1}
\end{equation}
(see \cite{bollobas01:_random_graph}). Here we describe a simple
method which can be used to separate sets of vertices of sublinear size.

We assume, for now, that $S$ and $T$ are each formed by $\gamma = \gamma(n)$ nodes with $\kappa(G(n,p)) \leq
\gamma < n/2$. Let $M = S \cup T$.
\begin{figure}[t]
  \centering
  \includegraphics[scale=0.3]{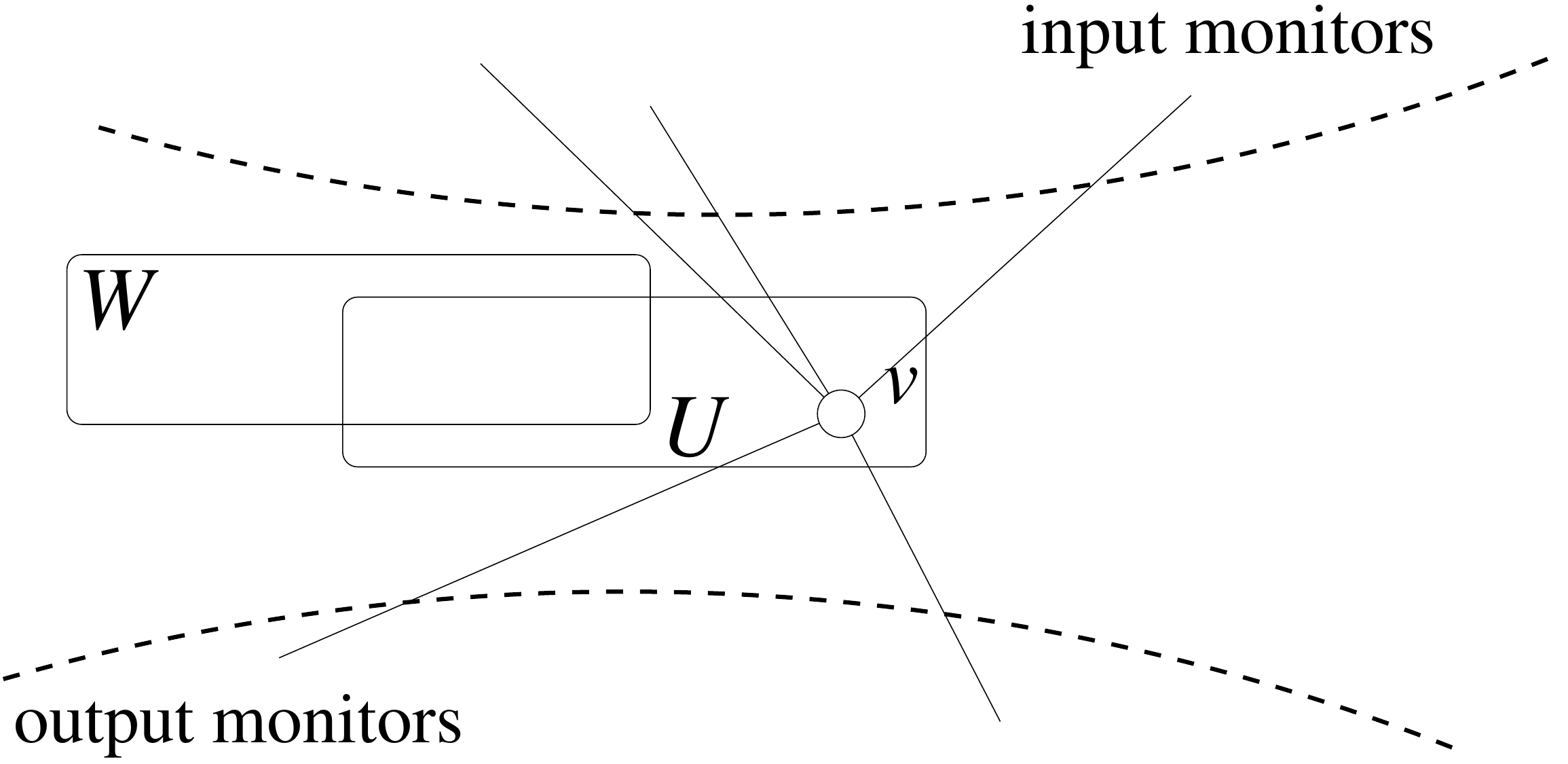}
  \caption{A node $v \in U \Delta W$ and a possible way to connect it
    to $S$ and $T$.}
  \label{fig:3}
\end{figure}
Let $U$ and $W$ be two
arbitrary subsets of $V \setminus M$ of size $k$. The probability that
$U$ and $W$ are separable is at least
the probability that an element $v$ of $U \Delta W$
(w.l.o.g. assume $v \in U \setminus W$) is directly connected to a node in
$S$ and to a node in $T$. This event has
probability $(1-(1-p)^{\gamma})^2.$ 
Hence the probability that $U$ and $W$ cannot be separated is at most
$1 - (1-(1-p)^{\gamma})^2 = 2(1-p)^{\gamma}- (1-p)^{2\gamma}$
and therefore the probability that some pair of sets $U$ and $W$ of
size $k$ (not
intersecting $M$) fail is at most
$2 {n - 2\gamma \choose k}^2 (1-p)^{\gamma}.$

\begin{theorem}
\label{gnp}
For fixed $p$ with $p \leq 1/2$, under the assumptions above about the way monitors are placed in
$G(n,p)$, the probability that $G(n,p)$ is not $k$-vertex separable is at most
$2 k {n \choose k}^2  \e^{(2k-\gamma) p}.$
\end{theorem}
\begin{proof}
The argument above works if both $U$ and $W$ contain no vertex
  in $M$. The presence of elements of vertices in $M$ in $U$ or $W$
  may affect the analysis in two ways. First $v$ could be in $M$ (say
  $v \in S$). In this case $U$ and $W$ are separable if $v$ is
  directly connected to a vertex in $T$. This happens with
  probability 
$(1-(1-p)^{\gamma}) > (1-(1-p)^{\gamma})^2.$
Second, $M$ might contain some elements of $U$ and $W$ different from
$v$. In the worst case when $v$ is trying to connect to $M$, it must
avoid at most $2k$ element of such set. 
There is at most $\sum_{h\leq k} {n \choose h}^2 \leq k {n \choose
  k}^2$ pairs of $U$ and $W$ of size at most $k$. Thus the probability
that $G(n,p)$ fails to be $k$-vertex separable is at most
$2 k {n \choose k}^2  (1-p)^{\gamma-2k}.$
and the result
 follows as $1-p \leq \e^{-p}$.
\qed\end{proof}

\subsection{Linear Separability in Erd\H{o}s-R\'enyi Graphs}

The argument above cannot be pushed all the way up to
$\kappa(G(n,p))$. When trying to separate vertex sets containing
$\Omega(n)$ vertices the problem is that these sets can form a large
part of $M$ and the existence of direct links from $v$ to
$S\setminus W$ and $T \setminus W$ is
not guaranteed with sufficiently high probability. However a different
argument allow us to prove the following:

\begin{theorem}
\label{gnp-tight}
For fixed $p$,  $\mu(G(n,p)=  \kappa(G(n,p))$ w.h.p. 
\end{theorem}

Full details of the proof are left to the final version of this paper,
but here is an informal explanation.
The upper bound follows immediately from (\ref{eq:1}) and  Theorem
\ref{thm:ub}.  For the lower bound we claim that the chance that two sets of size at
most $np$ are not vertex separable is small. To see this pick two 
sets $U$ and $W$, and remove, say, $W$. $G(n,p) \setminus W$ is still
a random graph on at least $n -  np$ vertices and constant edge
probability. Results in \cite{bollobas87:_algor_findin_hamil_paths_cycles_random_graph} imply that $G(n,p) \setminus W$ has a
Hamilton path starting at some $s \in S$ and ending at some $t \in T$ 
with probability at least $1-o(2^{-n})$ (and in fact
there is a fast polynomial time algorithm that finds one). Such
Hamilton path, by definition, contains a path from $S$ to $T$ passing
through $v \not \in W$, for every possible choice of $v$. This proves, w.h.p.,
the separability of sets of size up to $\kappa(G(n,p))$. Past such
value the construction in Theorem  \ref{thm:ub} applies.

\subsection{Random Regular Graphs}

A standard way to model random graphs with fixed vertex degrees is Bollobas' configuration model \cite{bollobas80:_probab_proof_asymp_formul_number_regul_graph}. There's
$n$ buckets, each
with $r$ free points. A random pairing of these free points has a
constant probability of not containing any pair containing two points
from the same bucket or two pairs containing points from just two
buckets. These configurations are in one-to-one correspondence with
$r$-regular $n$-vertex simple graphs.  Denote by ${\cal C}_{n,r}$ the
set of all configurations $C(n,r)$ on $n$ buckets each containing $r$ points,
and let $G(r$-reg$)$ be a random $r$-regular graph.

As before assume $|S| = |T| = \gamma$. The main result of this section is the
following:
\begin{theorem}
\label{reg}
 Let $r \geq 3$ be a fixed integer. $\mu(G(r$-reg$)= r$ w.h.p. 
\end{theorem}

The upper bound follows from Theorem \ref{thm:ub} and the
weil-known fact that random $r$-regular graphs are $r$-connected w.h.p.  The lower bound is a consequence of the following:

\begin{lemma}
\label{reg-lemma}
 Let $r \geq 3$ be a fixed integer.
Two sets $U$ and $W$ with $U, W \subseteq V(G(r$-reg$))$ and
 $\max(|U|,|W|) \leq k$ are separable w.h.p. if $k = r - o(1)$.
\end{lemma}
\begin{proof} In what follows we often use
  graph-theoretic terms, but we actually work with a random
  configuration $C(n,r)$.  Let $U$ and $W$ be two sets of $k$ buckets. For simplicity assume that
(the vertices corresponding to the elements of) both $U$ and $W$ are subsets of $V \setminus M$. The probability that
$U$ and $W$ can be separated 
is at least the probability that a (say) random element $v$ of $U \triangle W$
(w.l.o.g. $v \in U \setminus W$) is connected to $S$ by a path of
length at most $\ell_s$ and to $T$ by a path of length at most
$\ell_t$, neither of which ``touch'' $W$. 
Fig. \ref{fig:5} provides a simple example of the event under
consideration. The desired paths can be found using algorithm {\sc
  PathFinder} below that builds the paths and $C(n,r)$ at the same time.

\medskip

\fbox{\begin{minipage}{11.4cm}
    {\sc PathFinder$(v, \ell_s,\ell_t,W)$}

    \vspace{-2mm}
    
\begin{quote}
\begin{description}
  \item[{\sc SimplePaths}$(v,\ell_s,\ell_t,W)$.] Starting from $v$,
    build a simple path $p^s$ of
  length $\ell_s$ that
  avoids $W$. Similarly, starting
  from $v$, build a simple path $p^t$ of length $\ell_t$ that
  avoids $W$.
\item[{\sc RandomShooting}$(p^s,p^t)$.] Pair up all un-matched points
  in $p^s$ and $p^t$. 
\end{description}

\smallskip

Complete the configuration $C(n,r)$ by pairing up all remaining points.
\end{quote}
\end{minipage}}

\medskip

Sub-algorithm {\sc SimplePaths} can complete its constructions by pairing
points starting from elements of the bucket $v$ then choosing a random
un-matched point in a bucket $u$, then picking any other point $u$ and
then again a random un-matched point and so on, essentially simulating
two random walks RW$_s$ and RW$_t$ on the set of buckets. Note that the process may fail if at any point we re-visit a
  previously visited bucket or if we hit $W$ or even $M$. However the
  following can be proved easily.

\begin{claim}
{\rm RW}$_s$ and {\rm RW}$_t$ succeed w.h.p. provided $\ell_s, \ell_t \in o(n)$.
\end{claim}

As to {\sc RandomShooting},   the process
  succeeds if we manage to hit an element of $S$ from $p^s$ and an
  element of $T$ from $p^t$. 

  \begin{claim}
    {\sc RandomShooting}$(q_s,q_t,S,T)$ succeeds w.h.p. if    $\ell_s, \ell_t \in \omega(1)$.
  \end{claim}

Any un-matched point in $p^s$ or $p^t$ after {\sc SimplePaths} is
complete is called {\em
    useful}. Path $p^s$ (resp. $p^t$) contains $q_s = (r-2) \ell_s + 1$
  (resp $q_t = (r-2)  \ell_t + 1$) useful points. During the execution
  of {\sc RandomShooting}
a single useful point ``hits'' its target set, say $S$, with
  probability proportional to the cardinality of $S$. Hence the
  probability that none of the $q_s$ useful points hits $S$ is
  $(1-\frac{\gamma}{n})^{q_s}$ and the overall success probability is
$(1-(1-\frac{\gamma}{n})^{q_s})(1-(1-\frac{\gamma}{n})^{q_t}).$

\medskip

Back to the proof of Lemma \ref{reg-lemma} Set $\ell_s = \ell_t =
\ell$ and $q$ the common value of $q_s$ an $q_t$. The argument above implies that the success probability for $U$ and
$W$ is asymptotically approximately
$(1-(1-\frac{\gamma}{n})^{q})^2$
and the rest of the argument (and its conclusion) is very similar to
the $G(n,p)$ case (the final bound is slightly weaker, though). The chance that a random $r$-regular graph is not
$k$-vertex separable is at most
\[O(n^{2k}) \times  (1-(1-(1-\frac{\gamma}{n})^{q})^2) \leq  O(n^{2k}) \times
2(1-\frac{\gamma}{n})^q \leq   O(n^{2k}) \times
2 \e^{-\frac{\gamma}{n} q},\]
 which goes to zero as $n^{-C}$ provided $k \leq (r - o(1))\frac{\gamma
  \ell}{n \log n}$. The constraints on $\ell$ from the claims above
imply that that parameter can be traded-off agains $\gamma$ to achieve
optimal identifiability.
\qed\end{proof}
\begin{figure}[t]
  \centering
  \includegraphics[scale=0.2]{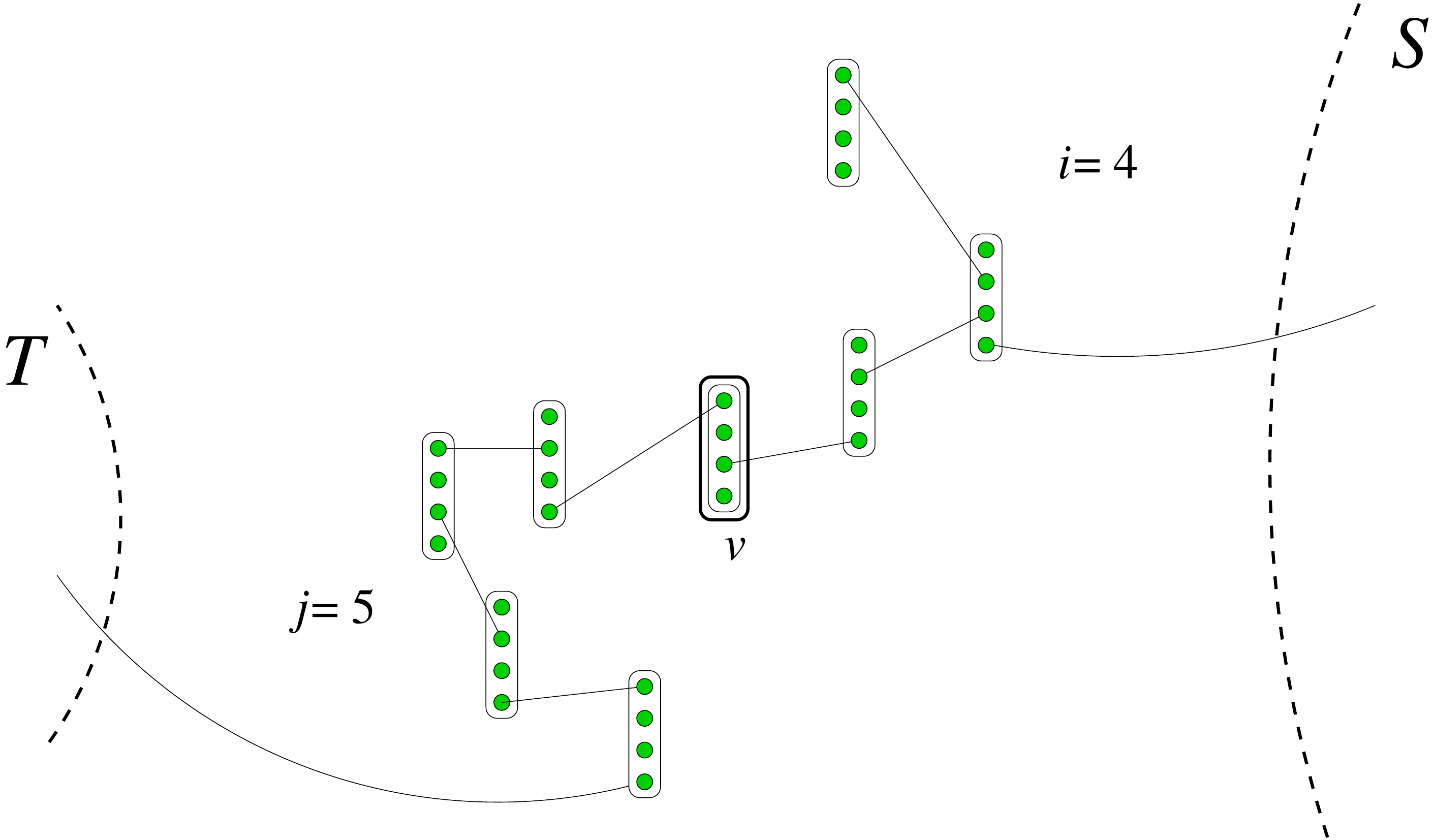}
  \caption{Assume $r=4$. The picture  represents a bucket
    (i.e. vertex) $v \in U \triangle W$ and two possible
    ``paths'' (sequences of independent edges such that consecutive
    elements involve points from the same bucket) of length 3 and 5, respectively connecting it to $S$ and $T$.
  \label{fig:5}}
\end{figure}

\newpage
\bibliographystyle{plain}

\newcommand{\etalchar}[1]{$^{#1}$}

\newpage
\section*{APPENDIX}

\begin{lemma}
Let $d\in \mathbb N^+$ and $n,\omega \in \mathbb N$, $n\geq 2$ and $\omega>2$. $\delta(\calH_{n, d, \omega})=d(\w-1)$.
\end{lemma}

\begin{proof}
Each node $u\in \calH_{n, d, \omega}$ has $\w-1$ edges for each one of
the possible directions the node 
is linked to. There can be at most $2d$ direction. Thus an internal
node has degree $2d(\w-1)$, the nodes on a border have degree
$(d+t)(\w-1)$ for some $t \in \{0, \ldots d-1\}$ and, in particular,
the nodes on the corners of the grid have degree $d(\w-1)$. Hence the minimal 
degree in $\calH_{n, d, \omega}$ is reached at the corner nodes and it is $d(\w-1)$. 
\end{proof}

By Lemma IV.4 of \cite{GR18} we have 
$\mu(\calH_{n, d, \omega})\leq d(\w-1)$. In the remainder of this section
we provide a tight lower bound.

Given a node $u\in \calH_{n, d, \omega}$, identified as $(x_1,...,x_d)\in \mathbb{Z}^d_{n}$, we define: 
$$SE(u)=\{(y_1,...,y_d): y_1\geq x_1 \vee...\vee y_d\geq x_d\}, \qquad \mbox{and}$$
$$NW(u)=\{(y_1,...,y_d): y_1\leq x_1 \vee...\vee y_d\leq
x_d\}. \qquad \quad$$

Note that $\calH_{n, d, \omega}$ has many more edges than the simple hypergrid $\calH_{n,d}$ (studied in \cite{GR18}). For non-degenerate monitor placement $(S,T)$, we place $d\w-1$ input monitors on the west and north borders of  $\calH_{n, d, \omega}$ and $d\w-1$ output monitors on the south and east borders of  $\calH_{n, d, \omega}$.

\begin{definition} 
Given a node $u \in \calH_{n, d, \omega}$ and $W$ a set of nodes in $\calH_{n, d, \omega}$. 
We say that a direction $X$ ($2d$ directions) is {\em $W$-saturated} on $u$ if moving from $u$ on direction $X$ there is, right after $u$,  a consecutive block of $\w-1$ nodes in $W$.
\end{definition}

\medskip

\paragraph{Proof of Theorem \ref{full-los-d}.}
We have to prove that for two node sets $U$, and $W$ of cardinality at most $d(\w-1)-1$, with $U\triangle W \not =\emptyset$ we can build a path from $S$ to $T$ touching exactly one of them.  Given a node $u \in U \setminus W$, let $S(u)=NW(u)$, the nodes in the North-West of $u$ and let $T(u)=SE(u)$, the nodes in South-East of $u$.
Notice that $(\overbrace{1\times \cdots \times 1}^{d}) \in S(u)$ and $(\overbrace{n\times \cdots \times n}^{d})\in T(u)$ and $S(u) \cap T(u) =\emptyset$.  Since $|S|=d\w-1>d(\w-1)$ and $|W|\leq d(\w-1)$, there is a node in $s \in S\setminus W$. Assume that $s=(1,...,1)$ (if $s\not =(1,...,1)$ is similar). Similarly for $T$, assume that $t=(n,...,n) \not \in W$.  Now let $u\in S(u)$. We show that there is a path $p_u$ in $S(u)$ from $s=(1,...,1)$ to $u$ not touching $W$. There is also a path $q_u$ in $T(u)$ from $u$ to $t=(n,...,n)$ not touching $W$. We prove the first one since they are the same.
By induction on $S(u)$. If $|S(u)|=1$, then $u=s$. Take the path $u$
itself and we have done. If $|S(u)|>1$, since $|W| \leq d(\w-1)-1$,
and since a direction is $W$-saturated only if a block of $\w-1$
consecutive elements of $W$ appear after $u$ on that direction, then
there is at least a direction $X$ between North and West which is not
$W$-saturated. Hence there is a node $u'\in S(u) \setminus W$ on
direction $X$ from $u$ at distance less than $\w$. Thus there is an
edge  $\{u',u\} \in \calH_{n, d, \omega}$. Since $S(u') \subset S(u)$,
by the inductive hypothesis we have a path $p_{u'}$ as required. Hence
the path $p_u=p_{u'},\{u',u\}$ is as required. Now we have found two
disjoint paths $p_u$ in $S(u)$ from $s$ to $u$  and $q_u$ in $T(u)$
form $u$ to $t$ not touching $W$. Their concatenation gives us a path
from $S$ to $T$ passing from $u$ and not touching $W$ and proves the
theorem.  \qed

\end{document}